\documentclass[12pt]{article}

\usepackage{amssymb,amsmath,amsthm,amscd,mathrsfs,amsbsy,amsfonts}
\usepackage{pstricks,pst-node}
\usepackage{amsbsy,young,colortbl,cite}

\usepackage{graphicx}

\usepackage{amsmath}

\usepackage{amssymb,amsthm}
\usepackage[english]{babel}
\renewcommand{\d}{\partial}

\newtheorem{proposition}{Proposition}
\newtheorem{lemma}{Lemma}

\newtheorem{theorem}{Theorem}

\newcommand{\gl}{M_N(\C)}

\renewcommand{\d}{\mathrm{d}}

\newcommand{\g}{\mathfrak{g}}

\newcommand\Zop{\mathbb{Z^{\mathrm{odd}}_+}}
\def\d{\partial}

\newcommand{\R}{\mathbb R}

\newcommand{\C}{\mathbb C}

\makeatletter
    
    \newcommand{\Rmnum}[1]{\expandafter\@slowromancap\romannumeral #1@}
  \makeatother

\def\({\left(}
\def\){\right)}
\def\[{\begin{eqnarray}}
\def\]{\end{eqnarray}}
\def\d{\partial}

\begin{document}

\title{Gauge transformation and symmetries  of the commutative multi-component BKP hierarchy}

\author{
\ Chuanzhong Li\\
 Department of Mathematics,  Ningbo University, Ningbo, 315211, China\\
Email:lichuanzhong@nbu.edu.cn\\
}

\date{}

\maketitle

\abstract{
In this paper, we defined a new multi-component BKP hierarchy which takes values in a commutative subalgebra  of $gl(N,\mathbb C)$. After this, we give the gauge transformation of this commutative multi-component BKP (CMBKP) hierarchy. Meanwhile we construct a new constrained CMBKP hierarchy which contains some new integrable systems including coupled  KdV equations under a certain reduction. After this, the quantum torus symmetry and quantum torus constraint on the tau function of the commutative multi-component BKP hierarchy will be constructed.
}\\


\tableofcontents

\section {Introduction}

The KP  and Toda lattice hierarchies
  as  completely integrable systems  have many important applications in mathematics and physics including the theory of Lie algebras' representation, orthogonal polynomials and  random
matrix model  \cite{bbkp,Toda,Todabook,UT,witten,dubrovin}. The KP hierarchy has many kinds of reduction or extension, for example the BKP, CKP hierarchies and so on. As  important sub-hierarchies of the KP hierarchy, the constrained KP (cKP) hierarchy, the constrained BKP  (cBKP) hierarchy and the constrained CKP(cCKP) hierarchy play an important role in the commutative integrable systems.

In \cite{tu}, the Virasoro symmetry and ASvM formula  of the BKP hierarchy were given.
In \cite{JHEP,gtf}, the gauge transformations of the BKP, CKP, constrained BKP and constrained CKP
hierarchies were constructed.
 In the paper \cite{BKP-DS}, we construct the generalized additional symmetries of the two-component BKP hierarchy
 and identify its algebraic structure.
 As a reduction of the two-component BKP hierarchy, the D type Drinfeld-Sokolov hierarchy was found to be a good differential model to derive a complete
 Block type infinite dimensional Lie algebra (also called Torus Lie algebra). About the Block algebra and its quantization (quantum torus algebra) related to integrable systems, we did a series of works in \cite{ourBlock}-\cite{torus}.
 In the paper \cite{NPB}, we  constructed the additional symmetries of the supersymmetric BKP hierarchy which constitute a B type $SW_{1+\infty}$ Lie algebra. Further we generalize the SBKP hierarchy to a supersymmetric two-component BKP hierarchy (S2BKP) hierarchy
and a new supersymmetric Drinfeld-Sokolov hierarchy of type D which  has a super Block type additional symmetry.

There is another kind of generalization of KP and Toda systems called multi-component KP \cite{kac,avanM} or multi-component Toda system which attracts more and more attention because of its widely use in many fields such as the fields of multiple orthogonal
polynomials and non-intersecting Brownian motions.
In \cite{commutBKP}, they considered a generalized multicomponent KP hierarchy which contains N independent generalized scalar KP hierarchies in particular by considering a commutative subalgebra of diagonal matrices. In \cite{spinor},  a formalism of multicomponent BKP hierarchies using elementary
geometry of spinors was developed by Kac and van de Leur. In \cite{ramond}, M. Ma\~{n}as, Luis Mart\'{i}nez Alonso  construct a relation between multicomponent BKP hierarchy and Lame
equations from Ramond fermions.
 The $\tau$ functions of a
$2N$-multicomponent KP hierarchy provide solutions of the $N$-multicomponent
two dimensional Toda hierarchy \cite{UT} which
 was considered from the point of view of the theory  of multiple matrix orthogonal polynomials, non-intersecting Brownian motions and matrix Riemann-Hilber problem \cite{manasInverse}-\cite{manas}. The multicomponent  Toda hierarchy in \cite{manasInverse} is a periodic reduction of bi-infinite matrix-formed  Toda hierarchy which contains matrix-formed Toda equation as the first flow equation. In \cite{EMTH}, we defined  the extended multi-component Toda hierarchy and its Sato theory.

In \cite{zuo}, a new hierarchy called as $Z_m$-KP hierarchy which take values in a maximal commutative subalgebra of $gl(m,\C)$ was constructed, meanwhile the relation between Frobenius manifold and dispersionless reduced $Z_m$-KP hierarchy was discussed. This inspired us to consider the Hirota quadratic equation of the commutative version of extended multi-component Toda hierarchy in \cite{EZTH} which might be useful in Frobenius manifold theory.

This paper is arranged as follows. In the next section we recall the factorization problem and construct the multicomponent $Z_N$-BKP hierarchy. In Section 2,
we will give the Lax equations of the commutative multicomponent BKP hierarchy.  In Section 3, multi-fold transformations of the CMBKP hierarchy will be constructed using the determinant technique in \cite{Hedeterminant,rogueHMB}. We construct a new constrained CMBKP hierarchy which contains some new integrable systems including a coupled commutative matrix $KdV$ equation in Section 4. In Section 5, the  quantum torus symmetry and quantum torus constraint on the tau function of the commutative multi-component BKP hierarchy will be constructed. Section 6 will be devoted to a short conclusions and discussions.

\section{ Lax equations of CMBKP hierarchy}

In this section we will use the factorization problem to derive  Lax equations.
We will consider the linear space of the complex $N\times N$
 matrix-valued function
$g:\R\rightarrow M_N(\C)$ with the derivative operator $\partial$.
Then the set $\g$ of Laurent
series in $\partial$ as an associative algebra is a
Lie algebra under the standard commutator.
This Lie algebra has the following important splitting
\begin{gather}\label{splitting}
\g=\g_+\oplus\g_-,
\end{gather}
where
\begin{align*}
  \g_+&=\Big\{\sum_{j\geq 0}X_j(x)\partial^j,\quad X_j(x)\in\gl\Big\},&
  \g_-&=\Big\{\sum_{j< 0}X_j(x)\partial^j,\quad X_j(x)\in\gl\Big\}.
\end{align*}

The splitting
\eqref{splitting} leads us to consider the following factorization of
$g\in G$
\begin{gather}\label{fac1}
g=g_-^{-1}\circ g_+, \quad g_\pm\in G_\pm
\end{gather}
where $G_\pm$ have $\g_\pm$ as their Lie algebras. $G_+$
is the set of invertible linear operators  of the
form $\sum_{j\geq 0}g_j(x)\partial^j$; while $G_-$ is the set of
invertible linear operators of the form
$1+\sum_{j<0}g_j(x)\partial^j$. This algebra has a maximal commutative subalgebra $Z_N=\C[\Gamma]/(\Gamma^N)$ and $\Gamma=(\delta_{i,j+1})_{ij}\in gl(N,\C).$
Denote $Z_N(\partial):=\g_c$, then
we have the following splitting
\begin{gather}\label{splittingc}
\g_c=\g_{c+}\oplus\g_{c-},
\end{gather}
where
\begin{align*}
  \g_{c+}&=\Big\{\sum_{j\geq 0}X_j(x)\partial^j,\quad X_j(x)\in Z_N\Big\},&
  \g_{c-}&=\Big\{\sum_{j< 0}X_j(x)\partial^j,\quad X_j(x)\in Z_N\Big\}.
\end{align*}

 We denote $``*"$ as a formal adjoint operation defined by $p^{*}=\sum(-1)^{i}\partial^{i}\circ p_i$ for an arbitrary
$Z_N$-valued pseudo-differential operator $p=\sum p_i\partial^{i}$, and $(fg)^{*}=g^{*}f^{*}$ for two operators $(f,g)$. Here $\circ$ means the multiplication of two operators.

Before the work, we list some identities, which will be used in the following sections:
\[A^{\ast}=A,\]
\[(AB)^{\ast}=BA,\]
\[(A\circ\partial\circ B)^{\ast}=-B\circ\partial\circ A,\]
where $A$ and $B$ are $N\times N$ $Z_N$-valued matrix functions.
The Lax operator of the   CMBKP hierarchy has  form
\begin{equation} \label{PhP}
L= \d+\sum_{i\ge1}u_i  \d^{-i},\end{equation}
where $u_i$ takes values in the commutative subalgebra $Z_N$.
And the operator must satisfy the following so-called B type condition
\[\label{laxBcondition}L^*=-\partial\circ L\circ\partial^{-1}.\]
The   CMBKP hierarchy is defined by the following
Lax equations:
\begin{align}\label{bkpLax}
& \d_{2k-1} L=[(B_{2k-1})_+,  L],\ \ B_{2k-1}=L^{2k-1}, \ \ k\geq 1.
\end{align}
One can write the operators $L$ in a dressing form as
\begin{equation} \label{PPh}
L=\Phi\circ \d\circ\Phi^{-1},
\end{equation}
where
\begin{align} \label{dreop}
\Phi=1+\sum_{i\ge 1}a_i \d^{-i},
\end{align}
 satisfy
\begin{equation}\label{phipsi}
\Phi^*=\d\circ\Phi^{-1}\circ\d^{-1}.
\end{equation}

We call eq.\eqref{phipsi} the B type condition of  the CMBKP hierarchy.
Given $L$, the dressing operators $\Phi$  are determined uniquely up to a multiplication to the
right by operators with
constant coefficients. The dressing operator $\Phi$ takes values in a B type commutative Volterra group in $G_-$. The   CMBKP  hierarchy \eqref{bkpLax}
can also be redefined as
\[
\frac{\d \Phi}{\d t_{2k-1}}=- (L^{2k-1})_-\circ\Phi,
\]
with $k\geq 1$.
In the the CMBKP hierarchy, we can derive an equation as following
\[\label{cmbkpeq}9v_{x,t_5}-5v_{t_3,t_3}+(v_{xxxxx}-5v_{xx,t_3}-15v_xv_{t_3}+15v_xv_{xxx}+15v_x^3)_x=0,\]
where $v=\int u_1 dx$ is in the $Z_N$ algebra.
We will call the eq.\eqref{cmbkpeq} the CMBKP equation.  When $N=2$, we can derive the following two-component CMBKP equation as
\[\label{cmbkpeq2}9w_{x,t_5}-5w_{t_3,t_3}+(w_{xxxxx}-5w_{xx,t_3}-15w_xw_{t_3}+15w_xw_{xxx}+15w_x^3)_x=0,\\
9z_{x,t_5}-5z_{t_3,t_3}+(z_{xxxxx}-5z_{xx,t_3}-15w_xz_{t_3}-15z_xw_{t_3}+15w_xz_{xxx}+15z_xw_{xxx}+45z_xw_x^2)_x=0,\]
where $v=w+z\Gamma.$
After freezing the $t_3$ flow, the CMBKP equation will be reduced to commutative two-component Sawada-Kotera(CMSK) equation as
\[\label{cmbkpeq2}9w_{x,t_5}+(w_{xxxxx}+15w_xw_{xxx}+15w_x^3)_x=0,\\
9z_{x,t_5}+(z_{xxxxx}+15w_xz_{xxx}+15z_xw_{xxx}+45z_xw_x^2)_x=0.\]
With the above preparation, it is time to construct gauge transformations for the CMBKP hierarchy in the next section.

\section{Gauge transformations of the CMBKP hierarchy}

In this section, we will consider the gauge transformation of the CMBKP hierarchy on the Lax operator

  \[\label{1darbouxL} L^{[1]}=\d+\sum_{i\ge1}U_i^{[1]}  \d^{-i}=W\circ L \circ W^{-1}, \]
where $W$ is the gauge transformation operator.
And $L^{[1]}$ should satisfy the B type condition
\[(L^{[1]})^*=-\partial \circ L^{[1]}\circ \partial^{-1},\]
which further implies
\[W^*=\partial\circ  W^{-1} \circ \partial^{-1}.\]

That means after the gauge transformation, the spectral problem about the $N\times N$ spectral matrix $\phi$ taking values in the commutative subalgebra $Z_N$ will preserve its form as

\[ L\cdot\phi=\lambda\phi,\ \ \
\frac{\d \phi}{\d t_n}=B_n\cdot\phi.\]

To keep the Lax pair of the CMBKP hierarchy invariant, i.e.,
  \begin{align}\label{bkpLax}
& \d_{t_n} L^{[1]}=[(B_n^{[1]})_+,  L^{[1]}],\ \ B_n^{[1]}=(L^{[1]})^{n}, \ \ n=1,3,5,\cdot \cdot \cdot,
\end{align}

the dressing operator $W$ should satisfy the following dressing equation
\begin{equation}\label{satisfy}
W_{t_{n}}=-W\circ (B_{n})_++(W\circ B_{n}\circ W^{-1})_+\circ W,\ \  n=1,3,5\cdot \cdot \cdot,
\end{equation}
where $W_{t_{n}}$ means the derivative of $W$ by $t_{n}.$

The evolutions of the eigenfunction $\phi$ and the adjoint eigenfunction $\psi$ of the CMBKP hierarchy are defined respectively by
\begin{equation}\label{rela}
\frac{\d \phi}{\d t_n}=B_n\cdot\phi, \frac{\d \psi}{\d t_n}=-(B_n)^*\cdot\psi,
\end{equation}
where $\phi=\phi(\lambda;t)$ and $\psi=\psi(\lambda;t)$ and $t=(t_1,t_3,t_5,...).$
To give the gauge transformation, we need the following lemma.

\begin{lemma}\label{lema}
The operator $B:=\sum_{n=0}^{\infty}b_n\partial^n(B:=\sum_{n=0}^{\infty}\partial^n\circ a_n)$ is a $Z_N$-valued differential operator and $f,g$ (short for $f(x),g(x)$) are two matrix functions taking values in the commutative subalgebra $Z_N$, following identities hold
\begin{equation}\label{Bneg}
(B\circ f \d^{-1}\circ  g)_-=(B\cdot f) \circ \d^{-1} \circ g,\ \ \ (f \d^{-1} \circ g\circ B)_-=f \d^{-1}\circ (B^*\cdot g).
\end{equation}
\end{lemma}

\begin{proof}
Here we only give the proof of the second equation of \eqref{Bneg} by direct calculation basing on the first equation of \eqref{Bneg}

\[\notag
(f \d^{-1} \circ g\circ  B)_-&=&(-B^{\ast}\circ  g\circ \partial^{-1}\circ  f)^{\ast}_-\\ \notag
&=&((-B^{\ast}\cdot g)\circ \partial^{-1} \circ f)^{\ast}\\ \notag
&=&\sum_{m=0}^{\infty}(-a_m((-\partial)^{m}\cdot g)\partial^{-1}\circ  f)^{\ast}\\ \notag
&=&\sum_{m=0}^{\infty}(-1)^{m}f\partial^{-1}\circ (\partial^{m}\cdot g)^{\ast} \circ a_m\\ \notag
&=&\sum_{m=0}^{\infty}f \d^{-1}(-1)^{m}\circ a_m(\partial^{m}\cdot g ) \\
&=&f \d^{-1} \circ (B^*\cdot g).\]

\end{proof}
\begin{lemma}
The operators $T_D=\phi\circ\d\circ\phi^{-1}$ and $T_I=\psi^{-1}\circ\partial^{-1}\circ\psi$ satisfy eq.\eqref{satisfy} ,
which implies $T_D T_I=\phi\circ\d\circ\phi^{-1}\circ\psi^{-1}\circ\partial^{-1}\circ\psi$ can also satisfy eq.\eqref{satisfy}.
\end{lemma}

Now, we will find out the gauge transformation operator $W$ of the CMBKP hierarchy.
Firstly, we consider the two operators
\[
T_D(\phi)=\phi\circ\partial\circ\phi^{-1}, T_I(\psi)=\psi^{-1}\circ\partial^{-1}\circ\psi
,\]
where $\phi$ and $\psi$ are  $N\times N$ matrix-valued eigenfunctions taking values in the commutative subalgebra $Z_N$.
Then we have
\[
(T^{-1}_D(\phi))^{\ast}=-T_I(\phi), (T^{-1}_I(\psi))^{\ast}=-T_D(\psi).
\]
We can easily get
\[T_D(\phi)\cdot \phi=0, (T^{-1}_I(\psi))^{\ast}\cdot \psi=0.\]

Similarly to the reference \cite{Hedeterminant}, we can consider two sets of matrix functions
${\{\phi^{(0)}_i,i=1,2,...n;\phi^{(0)}\}}$ and ${\{\psi^{(0)}_i,i=1,2\cdot \cdot \cdot n;
\psi^{(0)}\}}.$
For {\bf$T_D(\phi)=\phi\circ\partial\circ\phi^{-1}$}, we do iterations by the following two steps.
For the first step, we consider:
\[T^{(1)}_D=T^{(1)}_D(\phi^{(0)}_1)=\phi^{(0)}_1\circ \partial\circ(\phi^{(0)}_1)^{-1},\]
we define the rule of transformation under $T^{(1)}_D$ as
\[\phi^{(1)}=T^{(1)}_D(\phi^{(0)}_1)\cdot \phi^{(0)}, \psi^{(1)}=(T^{(1)}_D(\phi^{(0)}_1))^{*^{-1}}\cdot \psi^{(0)}=-T_I(\phi^{(0)}_1)\cdot\psi^{(0)},\]
\[\phi^{(1)}_i=T^{(1)}_D(\phi^{(0)}_1)\cdot \phi^{(0)}_i, \psi^{(1)}_i=(T^{(1)}_D(\phi^{(0)}_1))^{*^{-1}}\cdot \psi^{(0)}_i=-T_I(\phi^{(0)}_1)\cdot\psi^{(0)}_i,\]
where $i\geqslant 2$ for $\phi^{(1)}_i$
and \[\psi^{(1)}_i=-T_I(\phi^{(0)}_1)\cdot(\psi^{(0)}_i).\]

For the second step, we consider:
\[T^{(2)}_D=T^{(2)}_D(\phi^{(1)}_2)=\phi^{(1)}_2\circ \partial\circ(\phi^{(1)}_2)^{-1},\]

we define the rule of transformation under $T^{(2)}_D$ as
\[\phi^{(2)}=T^{(2)}_D(\phi^{(1)}_2)\cdot \phi^{(1)}, \psi^{(2)}=(T^{(2)}_D(\phi^{(1)}_2))^{*^{-1}}\cdot \psi^{(1)}=-T_I((\phi^{(1)}_2))\cdot\psi^{(1)},\]
\[\phi^{(2)}_i=T^{(2)}_D(\phi^{(1)}_2)\cdot \phi^{(1)}_i, \psi^{(2)}_i=(T^{(2)}_D(\phi^{(1)}_2))^{*^{-1}}\cdot \psi^{(1)}_i=-T_I((\phi^{(1)}_2))\cdot\psi^{(1)}_i,\]
where $i\geqslant 3$ for $\phi^{(2)}_i$
and \[\psi^{(2)}_i=-T_I((\phi^{(1)}_2))\cdot(\psi^{(1)}_i).\]

For $T_I(\psi)=\psi^{-1}\circ\partial^{-1}\circ\psi$, it obeys  the following iterated rule:

For the first step, we consider:
\[T^{(1)}_I=T^{(1)}_I(\psi^{(0)}_1)=(\psi^{(0)}_1)^{-1}\circ \partial^{-1}\circ(\psi^{(0)}_1),\]

\[\phi^{(1)}=T^{(1)}_I(\psi^{(0)}_1)\cdot \phi^{(0)}, \psi^{(1)}=(T^{(1)}_I(\psi^{(0)}_1))^{*^{-1}}\cdot \psi^{(0)}=-T_D((\psi^{(0)}_1))\cdot\psi^{(0)},\]
\[\phi^{(1)}_i=T^{(1)}_I(\psi^{(0)}_1)\cdot \phi^{(0)}_i, \psi^{(1)}_i=(T^{(1)}_I(\psi^{(0)}_1))^{*^{-1}}\cdot \psi^{(0)}_i=-T_D((\psi^{(0)}_1))\cdot\psi^{(0)}_i,\]
where $i\geqslant 2$ for $\psi^{(1)}_i$
and \[\psi^{(1)}_i=-T_D((\psi^{(0)}_1))\cdot(\psi^{(0)}_i).\]

For the second step, we consider:
\[T^{(2)}_I=T^{(2)}_I(\psi^{(1)}_2)=(\psi^{(1)}_2)^{-1}\circ \partial^{-1}\circ(\psi^{(1)}_2),\]

\[\phi^{(2)}=T^{(2)}_I(\psi^{(1)}_2)\cdot \phi^{(1)}, \psi^{(2)}=(T^{(2)}_I(\psi^{(1)}_2))^{*^{-1}}\cdot \psi^{(1)}=-T_D(\psi^{(1)}_2)\cdot\psi^{(1)},\]
\[\phi^{(2)}_i=T^{(2)}_I(\psi^{(1)}_2)\cdot \phi^{(1)}_i, \psi^{(2)}_i=(T^{(2)}_I(\psi^{(1)}_2))^{*^{-1}}\cdot \psi^{(1)}_i=-T_D(\psi^{(1)}_2)\cdot\psi^{(1)}_i,\]
where $i\geqslant 3$ for $\psi^{(1)}_i$
and \[\psi^{(2)}_i=-T_D(\psi^{(1)}_2)\cdot(\psi^{(1)}_i).\]

It is obvious that a single step of the operator $T_D$ or $I_I$ can not keep the restriction of the B type condition,
we use
\[
W_1=T_{1+1}=T_I(\psi^{(1)}_1)\circ  T_D(\phi^{(0)}_1)
,\]
as the gauge transformation operator and we have $L^{[1]}=W_1LW^{-1}_1$. Let us check whether it satisfies the required constraint
\[(L^{[1]})^*=-\partial L^{[1]}\partial^{-1},\]
We can calculate
\[\notag
(L^{[1]})^*&=&((\psi^{(1)}_1)^{-1}\circ \partial^{-1}\circ\psi^{(1)}_1\circ\phi^{(0)}_1\circ\partial\circ(\phi^{(0)}_1)^{-1}\circ\\ \notag
&&L \circ\phi^{(0)}_1\circ\partial^{-1}\circ(\phi^{(0)}_1)^{-1}\circ(\psi^{(1)}_1)^{-1}\circ\partial\circ\psi^{(1)}_1)^{\ast}\\ \notag
&=&-(\psi^{(1)}_1)\circ \partial\circ((\psi^{(1)}_1))^{-1}(\phi^{(0)}_1)^{-1}\circ\partial^{-1}\circ(\phi^{(0)}_1)\circ\partial\circ L\\  &&\circ\partial^{-1}\circ(\phi^{(0)}_1)^{-1}\circ\partial\circ(\phi^{(0)}_1)\circ(\psi^{(1)}_1)\circ\partial^{-1}\circ((\psi^{(1)}_1))^{-1},\]
and
\[-\partial L^{[1]}\partial^{-1}=-\partial\circ((\psi^{(1)}_1))^{-1}\circ\partial^{-1}\circ\psi^{(1)}_1\circ\phi^{(0)}_1\circ\partial\circ(\phi^{(0)}_1)^{-1}\\ \notag
\circ L
\circ\phi^{(0)}_1\circ\partial^{-1}\circ(\phi^{(0)}_1)^{-1}\circ(\psi^{(1)}_1)^{-1}\circ\partial\circ\psi^{(1)}_1\circ\partial^{-1},\]
which means in order to keep the constraint $(L^{[1]})^*=-\partial L^{[1]}\partial^{-1}$, $T$ should satisfy the following equation:
\begin{equation}\label{opera}
T_D(\psi^{(1)}_1)T_I(\phi^{(0)}_1)\circ \partial=\partial \circ T_I(\psi^{(1)}_1)T_D(\phi^{(0)}_1),
\end{equation}
where $\psi^{(1)}_1=-(\phi^{(0)}_1)^{-1}\int(\phi^{(0)}_1)\psi^{(0)}_1$ and $\int$ means the integral about spatial variable $x$.

Then we can acquire the following theorem  because the CMBKP hierarchy  takes values in a commutative subalgebra just like the case when $N=1$, i.e. the case of the original BKP hierarchy..

\begin{theorem}
The B type condition of the CMBKP hierarchy implies $\psi^{(0)}_1$ and $\phi^{(0)}_1$ have the following relation:
 \[\label{reduc2}\psi^{(0)}_1=\phi^{(0)}_{1,x},\]
 \end{theorem}
The B-type reduction of $L$ guarantee that there exists at least one
solution ($\phi;\psi$ ) which satisfies eq.\eqref{reduc2}.
 In fact the above theorem can be generalized to the case of the $gl(N,\mathbb C)$-valued multicomponent BKP hierarchy which is not commutative.

The B type condition of the $gl(N,\mathbb C)$-valued multicomponent BKP hierarchy implies noncommutative matrices $\psi^{(0)}_1$ and $\phi^{(0)}_1$ have the following relation:
\begin{equation}\label{relation1}
((\phi^{(0)}_1)^T)^{-2}(\phi^{(0)}_1)^T_x\int(\phi^{(0)}_1)^T\psi^{(0)}_1-\psi^{(0)}_1\\
-((\phi^{(0)}_1)^T)^{-1}(\int(\phi^{(0)}_1)^T\psi^{(0)}_1)\phi^{(0)}_{1,x}(\phi^{(0)}_1)^{-1}
+(\phi^{(0)}_1)^T_x=0,
\end{equation}
where $T$ means the transpose of matrices.

The proof of the eq.\eqref{relation1} will be skipped here because the focus of this paper is about the CMBKP hierarchy. A thorough study on the  $gl(N,\mathbb C)$-valued multicomponent BKP hierarchy will be contained in another work of ours recently.

{\bf Remark:} From eq.\eqref{reduc2} to eq.\eqref{relation1}, one can see clearly the difference of the BKP systems from $Z_N$ to $gl(N,\mathbb C)$.

In order to keep the B type restriction of the Lax operator of the CMBKP hierarchy, we do iterations of the gauge transformation $W_n=T_{n+n}$. In particular,
\[
W_2=T_{2+2}=T_I(\psi^{(3)}_2)\circ T_D(\phi^{(2)}_2)\circ  T_I(\psi^{(1)}_1)\circ T_D(\phi^{(0)}_1)
,\]
\[
W_n=T_{n+n}=T_I(\psi^{(2n-1)}_n)\circ T_D(\phi^{(2n-2)}_n)\circ ... \circ T_I(\psi^{(1)}_1)\circ T_D(\phi^{(0)}_1)
,\]
where
$\psi^{(2n-1)}_i=-T_I((\phi^{(2n-2)}_n))\cdot(\psi^{(2n-2)}_i),\ \ \psi^{(i)}_n=(\phi^{(i)}_n)_x.$
It can be easily checked that $W_n\cdot\phi^{(0)}_i\mid_{i\leq n}=0,(W_n^{-1})^{\ast}\cdot(\psi^{(0)}_i)\mid_{i\leq n}=0$.
The relations $\psi^{(i)}_n=(\phi^{(i)}_n)_x,\ \ n=1,2...$ can keep the dressing procedures $W_n=T_{n+n},\ \ n=1,2...$ always preserving the B type condition of new Lax operators $L^{[n]}$. This is similar as the case of the BKP hierarchy in \cite{JHEP}.

We denote $t=(t_1,t_3,t_5,\dots)$ and introduce
the  $Z_N$-valued wave function as
\begin{align}\label{wavef}
w(t; z)=\Phi \cdot e^{\xi(t;z)},
\end{align}
where the  function $\xi$ is defined as $\xi(t;
z)=\sum_{k\in\Zop} t_k z^k$. It is easy to see
\[
L\,w(t;z)=z w(t;z),\ \ \frac{\d w}{\d t_{2n+1}}=L^{2n+1}_{+}w.
\]

The $Z_N$-valued tau function $\tau$ of the   CMBKP hierarchy can be defined in form of the wave functions as
\begin{align}\label{wtau}
w(t,z)=\frac{\tau(t-2[z^{-1}])}{\tau(t)}
e^{\xi(t;z)},
\end{align}
where $[z]=\left(z,z^3/3,z^5/5,\dots\right)$.

The generating functions of n-step $T_D$ and n-step $T_I$ are denoted as $(\phi_1,\dots,\phi_{n-1},\phi_n)$	
 and 	$(\psi_1,\dots,\psi_{n-1},\psi_n)$	
 in order respectively.
 The generating functions have the following B type constraint
 \begin{equation}\label{relation1w}
\psi_i=(\phi_i)_x.
\end{equation}
Using the above gauge transformation, we can derive the gauge transformation on the tau function of the CMBKP hierarchy as
\[\tau^{(n+n)}=GW_{n,n}(\psi_n,\psi_{n-1},\dots , \psi_1;\phi_1,\dots,\phi_{n-1},\phi_n)\tau,\]
where the generalized Wronskian $GW_{k,n}$ is defined as \cite{gtf}
\[&&GW_{k,n}(g_k,g_{k-1},\dots, g_1;f_1, f_2,\dots,f_n)\\
&&=\left|\begin{matrix}\int g_kf_1&\int g_kf_2&\int g_kf_3&\dots&\int g_kf_n\\
\int g_{k-1}f_1&\int g_{k-1}f_2&\int g_{k-1}f_3&\dots&\int g_{k-1}f_n\\
\vdots&\vdots&\vdots&\dots&\vdots\\
\int g_1f_1&\int g_1f_2&\int g_1f_3&\dots&\int g_1f_n\\
f_1&f_2&f_3&\dots&f_n\\
f_{1x}&f_{2x}&f_{3x}&\dots&f_{nx}\\
\vdots&\vdots&\vdots&\dots&\vdots\\
(f_1)^{(n-k-1)}&(f_2)^{(n-k-1)}&(f_3)^{(n-k-1)}&\dots&(f_n)^{(n-k-1)}
\end{matrix}\right|.\]
When $k=0$, the generalized Wronskian $GW_{0,n}$ will be reduced to the ordinary Wronskian.
Now, we will only give the first gauge transformation of the CMBKP hierarchy included in the following proposition.

\begin{proposition}
If the eigenfunction $\phi$ and the adjoint eigenfunction $\psi$ satisfy the eq.\eqref{rela},
the one-fold gauge transformation operator of the CMBKP hierarchy
\begin{equation}
W_1=(\psi^{(1)}_1)^{-1}\circ\partial^{-1}\circ\psi^{(1)}_1\circ\phi^{(0)}_1\circ\partial\circ(\phi^{(0)}_1)^{-1},
\end{equation}
satisfies $W_1\phi_1^{(0)}=0$ and $(W^{-1}_1)^{\ast} (\psi^{(0)}_1)=0$.
$W_1$ will generater new solutions $U_{i}^{[1]}$ from seed solutions
$U_{i}$. To see it clearly, here we only give the transformations of first two dynamic functions

\begin{equation}\label{newo}
U_{1}^{[1]}=U_1+2(\ln\phi^{(0)}_1 )_{xx}
,
\end{equation}
\begin{equation}\label{newt}
U_{2}^{[1]}=U_2+4(\ln\phi^{(0)}_1 )_{xx}(\ln\phi^{(0)}_1 )_{x} -2(\frac{\psi^{(0)}_{1,x}}{ \phi^{(0)}_1})_x
.
\end{equation}

\end{proposition}
If we suppose $U_1=\begin{bmatrix}\alpha&0\\ \beta&\alpha\end{bmatrix}$,$U_2=\begin{bmatrix}\gamma&0\\ \eta&\gamma\end{bmatrix}$, $\phi^{(0)}_1 =\begin{bmatrix}\phi_0&0\\ \phi_1&\phi_0\end{bmatrix}$, then we can derive the explicit transformation as

\[
\alpha^{[1]}&=&\alpha+2(\ln\phi_{0} )_{xx},\ \
\\
 \beta^{[1]}&=&\beta+2(\frac{\phi_1}{\phi_0})_{xx},
\\
\gamma^{[1]}&=&\gamma+4(\ln\phi_{0} )_{xx}(\ln\phi_{0} )_{x} -2(\frac{\phi_{0xx}}{ \phi_{0}})_x,
\\
\eta^{[1]}&=&\eta+4(\ln\phi_{0})_{xx}(\frac{\phi_1}{\phi_0} )_{x}+4(\frac{\phi_1}{\phi_0})_{xx}(\ln\phi_{0})_{x} -2(\frac{\phi_{1xx}}{ \phi_{0}}-\frac{\phi_{0xx}\phi_1}{ \phi_{0}^2})_x.
\]
In the calculation, the identity $\ln{\begin{bmatrix}\phi_0&0\\ \phi_1&\phi_0\end{bmatrix}}=\begin{bmatrix}\ln{\phi_0}&0\\ \frac{\phi_1}{\phi_0}&\ln{\phi_0}\end{bmatrix}$
is used.

For the case $N=1$, $W_1$ will generater new solutions of the BKP hierarchy from seed solutions.

\section{Constrained CMBKP hierarchy}

In this section, we will consider the operator of the constrained CMBKP(cCMBKP) hierarchy as
\[L=\partial+u\partial^{-1} v_{x}-v\partial^{-1} u_{x},
\]
where $u$ and $v$ are $N\times N$ matrix functions taking values in $Z_N$.
Here $u,v$ should satisfy the following Sato equation
\[u_{t_{2n-1}}=B_{2n-1}\cdot u,\ v_{t_{2n-1}}=B_{2n-1}\cdot v.\]
Because the B type condition eq.\eqref{laxBcondition}, one can prove that the $\partial^0$ term does not exist in $B_{2n-1}$ as mentioned in \cite{bbkp}. That means that $u=v=1$ is a trivial solution.

Suppose $u=q+p\Gamma,v=r+s\Gamma$, and consider the case when $N=2$, i.e. \[u=\begin{bmatrix}q&0\\ p&q\end{bmatrix},\ \ v=\begin{bmatrix}r&0\\ s&r\end{bmatrix}.\]

Then we can further derive the following coupled equations
\[\notag q_{t_3}&=&q_{3x}+3(qr_{x}-rq_{x})q_x \\
\notag p_{t_3}&=&p_{3x}+3[(qr_{x}-rq_{x})p_x+(pr_x+qs_x-sq_x-rp_x)q_x]\\
\notag r_{t_3}&=&r_{3x}+3(qr_{x}-rq_{x})r_x\\
\notag s_{t_3}&=&s_{3x}+3[(qr_{x}-rq_{x})s_x+(pr_x+qs_x-sq_x-rp_x)r_x].\]

If $q=r, p=s$, we can derive the following trivial equations
\[\notag q_{t_3}&=&q_{3x}, \\
\notag p_{t_3}&=&p_{3x}.\]

If $r=s=1$, we can derive the following coupled matrix KdV-like equation
\[\notag q_{t_3}&=&q_{3x}-3q_{x}^2, \\
\notag p_{t_3}&=&p_{3x}-6q_{x}p_x-3q_x^2.\]

Similarly as \cite{gtf}, we can derive the new solutions generated from the seed solution $q,r$
\[u^{(n+n)} =\frac{GW_{n,n+1}(\psi^{(n-1)},\psi^{(n-2)},\dots , \psi^{(1)},u_x;u,\phi^{(1)},\dots,\phi^{(n-1)},\phi^{(n)})}
{GW_{n,n}(\psi^{(n-1)},\psi^{(n-2)},\dots , \psi^{(1)},u_x;u,\phi^{(1)},\dots,\phi^{(n-2)},\phi^{(n-1)})}
,\]

\[v^{(n+n)} =\frac{(-1)^nGW_{n-1,n}(\psi^{(n-2)},\psi^{(n-3)},\dots , \psi^{(1)},v_x;v,\phi^{(1)},\dots,\phi^{(n-2)},\phi^{(n-1)})}
{GW_{n,n}(\psi^{(n-1)},\psi^{(n-2)},\dots , \psi^{(1)},u_x;u,\phi^{(1)},\dots,\phi^{(n-2)},\phi^{(n-1)})},\]
where $\phi^{(j)}=L^ju$ and $(\phi^{(j)},\psi^{(j)})$ have the same relation as eq.\eqref{relation1w}.
Also the iteration on the constrained tau functions $\tau_c$ of the constrained CMBKP hierarchy as
\[\tau_c^{(n+n)}=GW_{n,n}(\psi^{(n-1)},\psi^{(n-2)},\dots , \psi^{(1)},u_x;u,\phi^{(1)},\dots,\phi^{(n-2)},\phi^{(n-1)})\tau_c.\]
In the above process of calculations, all the elements in above Wronskians must keep being always written in terms of $\Gamma$. In this way, one can keep the new solutions $u^{(n+n)}, v^{(n+n)} $ take values in the algebra $Z_N$.

\section{ Quantum torus constraint of CMBKP hierarchy}

In this section, we will focus on the quantum torus symmetry of the CMBKP hierarchy.
Firstly we define the   operator $\Gamma_B$ and the $Z_N$-valued  Orlov-Shulman's  operator $M$  as
 \begin{equation}
\Gamma_B=\sum_{i\in \Zop }it_i\partial^{i-1},\ \ M= \Phi \Gamma_B \Phi^{-1}.
\end{equation}
The Lax operator $L$ and the  $Z_N$-valued Orlov-Shulman's $M$ operator satisfy the following canonical relation
\[[L,M]=1.\]

With the above preparation, it is time to  construct additional symmetries for the  CMBKP hierarchy in the next part.
Then it is easy to get  that the operator $M$  satisfy

\begin{equation}
[L, M]=1,  \
M w(z)=\d_z w(z);
\end{equation}
\begin{equation}\label{bkpMt}
\frac{\d M}{\d t_k}=[(L^k)_+,M],\ \ k\in\Zop.
\end{equation}

Given any pair of integers $(m,n)$ with $m,n\ge0$, we will introduce the following  $Z_N$-valued operator $B_{m n}$
\begin{align}\label{defBoperator}
B_{m n}=M^mL^{n}-(-1)^{n} L^{n-1}M^{m}L.
\end{align}

For any $Z_N$-valued  operator $B_{m n}$ in \eqref{defBoperator}, one has
\begin{align}\label{Bflow}
&\frac{\d B_{m n}}{\d t_k}=[(L^k)_+, B_{m n}],  \ k\in\Zop.
\end{align}

Using
\[
 \Phi^*=\d\Phi^{-1} \d^{-1},\ \ \Gamma_B^*=\Gamma_B;\]
the    $Z_N$-valued operator $M$ satisfies the following identity,
\[\label{MBproperty}
M^*
=\d L^{-1}ML\d^{-1}.\]
 It is easy to check that the  $Z_N$-valued operator
$B_{m n}$  satisfy the B type condition, namely
\begin{equation}\label{btypeB}
B_{m n}^*=-\d  B_{m n} \d^{-1}.
\end{equation}

Now we will denote the operator $D_{m n}$ as
\begin{equation}
D_{m n}:=e^{mM}q^{nL}-L^{-1}q^{-nL}e^{mM}L.
\end{equation}

Using eq. \eqref{btypeB}, the B type  property of $D_{m n}$ can be derived as
\begin{eqnarray*}D_{m n}^*
&=&-\d  D_{m n} \d^{-1}.
\end{eqnarray*}
Therefore we get the following important B type condition which the $Z_N$-valued operator $D_{m n}$ satisfies
\begin{equation}
D_{m n}^*=-\d  D_{m n} \d^{-1}.
\end{equation}

Then basing on a quantum parameter $q$, the additional flows for the time variable $t_{m,n},t_{m,n}^*$ are
defined as follows
\begin{equation}
\dfrac{\partial \Phi}{\partial t_{m,n}}=-(B_{m n})_-\Phi,\ \
 \dfrac{\partial \Phi}{\partial t^*_{m,n}}=-(D_{m n})_-\Phi,
\end{equation}
or equivalently rewritten as

\begin{equation}
\dfrac{\partial L}{\partial t_{m,n}}=-[(B_{m n})_-,L], \qquad
\dfrac{\partial M}{\partial t^*_{m,n}}=-[(D_{m n})_-,M].
\end{equation}

 Generally, one can also derive
\begin{equation}\label{bkpMLK}
\partial_{t^*_{l,k}}(D_{m n})=[-(D_{l k})_-,D_{m n}].
\end{equation}

Using the similar proof  as the BKP hierarchy  in \cite{torus},
the additional flows of $\partial_{t^*_{l,k}}$ can be proved to be  symmetries of the  CMBKP hierarchy, i.e. they commute with all $\partial_{t_n}$ flows of the   CMBKP hierarchy.

The additional flows $\partial_{t_{l,k}}$ of the  CMBKP hierarchy form the
$W_{\infty}$ algebra similarly as \cite{tu} which is about the BKP hierarchy.

Now it is time to identity the algebraic structure of the
additional $t_{l,k}^*$ flows of the  CMBKP hierarchy.
\begin{theorem}\label{bkpalg}
The additional flows $\partial_{t_{l,k}^*}$ of the  CMBKP hierarchy form the
positive half of quantum torus algebra, i.e.,
\begin{equation}
[\partial_{t^*_{n,m}},\partial_{t^*_{l,k}}]=(q^{ml}-q^{nk})\partial_{t^*_{n+l,m+k}},\ \ n,m,l,k\geq 0.
\end{equation}

\end{theorem}

{\bf Remark:}
The $t^*_{l,k}$ additional flows constitute a nice quantum torus algebra because its basing on a commutative algebra. This is different from the multicomponent BKP whose additional symmetry constitute multi-fold quantum torus algebra \cite{dmBKPtorus}.

 Next, similar to the KP and BKP hierarchy \cite{torus},
we will consider the quantum torus constraint on the $Z_N$-valued  tau function of the CMBKP hierarchy.

Similar as  \cite{torus}, one has shown that
\begin{eqnarray}
\partial_{t_{p,s}}\log w=(e^{\tilde\eta}-1)\frac{\frac{Z_s^{(p+1)}}{p+1}(\tau)
}{\tau},
\end{eqnarray}
where \[\tilde\eta=\sum_{i\in \Zop}\frac{\lambda^{-i}}{i}\frac{\partial}{\partial t_i},\]
and
$Z_s^{(p+1)}$ is the generator of the  $W^B_{\infty}$ algebra.
Then with the help of rewriting  the quantum torus flow $\partial_{t^*_{l,k}}$ in terms of the $\partial_{t_{p,s}}$ flows
\begin{eqnarray*}
\partial_{t^*_{l,k}}&=&\sum_{p,s=0}^{\infty}\frac{l^p(k\log q)^s}{p!s!}\partial_{t_{p,s}},
\end{eqnarray*}
and denoting \[L^B_{l,k}:=\sum_{p,s=0}^{\infty}\frac{l^p(k\log q)^s}{p!s!}\frac{Z_s^{(p+1)}}{p+1},\]
 the quantum torus constraint on the $Z_N$-valued  wave function $w$, i.e.
\begin{eqnarray}
\partial_{t^*_{l,k}}w&=&0,
\end{eqnarray}
will lead to the quantum torus constraint on the $Z_N$-valued  tau function of the CMBKP hierarchy
\[L^B_{l,k}\tau=c,\] \label{constraintontauBKP}
where $c$ is a constant.

\section{Conclusions and Discussions}
In this paper, we defined a new multi-component BKP hierarchy which takes values in a commutative subalgebra  of $gl(N,\mathbb C)$. After this, we give the gauge transformation of the commutative multi-component BKP hierarchy. Meanwhile we construct a new constrained CMBKP hierarchy which contains some integrable systems including coupled matrix $KdV$ equations under a certain reduction. After this, the quantum torus symmetry and quantum torus constraint of the commutative multi-component BKP hierarchy are constructed. We are looking forward to the possible application of the quantum torus constraint in the topological field theory and enumerate geometry. For the importance of the BKP hierarchy in representation theory and mathematical physics, what is the application of the commutative multi-component BKP hierarchy in other theories such as Frobenius manifold is an interesting question.

{\bf {Acknowledgements:}}
  This work is supported by the National Natural Science Foundation of China under Grant No. 11201251, 11571192,
 the Zhejiang Provincial Natural Science Foundation under Grant No. LY15A010004, LY12A01007, the Natural Science Foundation of Ningbo under Grant No. 2015A610157 and the K. C. Wong Magna Fund in
Ningbo University.

%%%%%%%%%%%%%%%%% References  %%%%%%%%%%%%%%%%%%%%%%%%%%%%%%%%%%%%%%%

\end{document}